\documentclass[twocolumn,10pt]{IEEEtran}

\makeatletter

\newcommand{\Rmnum}[1]{\expandafter\@slowromancap\romannumeral #1@}
\makeatother
\usepackage{amssymb}
\usepackage{amsfonts}
\usepackage{amsthm}
\usepackage{amsmath}
\usepackage{cite}
\usepackage[dvips]{graphicx}
\usepackage{caption}
\usepackage{subcaption}
\usepackage{algorithm}
\usepackage{algorithmic}

\newtheorem{lemma}{Lemma}

\begin{document}

\bibliographystyle{ieeetr}

\title{Computing Resource Allocation in Three-Tier IoT Fog Networks: a Joint Optimization Approach Combining Stackelberg Game and Matching}

\author{\IEEEauthorblockN{ Huaqing Zhang\IEEEauthorrefmark{1}, Yong Xiao\IEEEauthorrefmark{6}, Shengrong Bu\IEEEauthorrefmark{2}, Dusit Niyato\IEEEauthorrefmark{4}, Richard Yu\IEEEauthorrefmark{3}, and Zhu Han\IEEEauthorrefmark{1}}\\
\IEEEauthorblockA{\IEEEauthorrefmark{1}Department of Electrical and Computer Engineering, University of Houston, Houston, TX, USA.}\\
\IEEEauthorblockA{\IEEEauthorrefmark{6}Department of Electrical and Computer Engineering at the University of Arizona, Tucson, Arizona, USA.}\\
\IEEEauthorblockA{\IEEEauthorrefmark{2}School of Engineering, University of Glasgow, UK.}\\
\IEEEauthorblockA{\IEEEauthorrefmark{4}School of Computer Science and Engineering, Nanyang Technological University, Singapore.}\\
\IEEEauthorblockA{\IEEEauthorrefmark{3}Department of Systems and Computer Engineering, Carleton University, Ottawa, ON, Canada.}}

%
%
%
%
%

\maketitle

\begin{abstract}

Fog computing is a promising architecture to provide economic and low latency data services for future Internet of things (IoT)-based network systems. It relies on a set of low-power fog nodes that are close to the end users to offload the services originally targeting at cloud data centers. In this paper, we consider a specific fog computing network consisting of a set of data service operators (DSOs) each of which controls a set of fog nodes to provide the required data service to a set of data service subscribers (DSSs). How to allocate the limited computing resources of fog nodes (FNs) to all the DSSs to achieve an optimal and stable performance is an important problem. In this paper, we propose a joint optimization framework for all FNs, DSOs and DSSs to achieve the optimal resource allocation schemes in a distributed fashion. In the framework, we first formulate a Stackelberg game to analyze the pricing problem for the DSOs as well as the resource allocation problem for the DSSs. Under the scenarios that the DSOs can know the expected amount of resource purchased by the DSSs, a many-to-many matching game is applied to investigate the pairing problem between DSOs and FNs. Finally, within the same DSO, we apply another layer of many-to-many matching between each of the paired FNs and serving DSSs to solve the FN-DSS pairing problem. Simulation results show that our proposed framework can significantly improve the performance of the IoT-based network systems. 

\end{abstract}


\section{Introduction}

With the rapid development of Internet of things (IoT), the number of connected devices has increased at a unprecedented speed\cite{JGubbi001}. 
It is known that analyzing the big data generated by all kinds of IoT devices requires a large amount of computing resources. In order to meet the demand of the data computing services, a large number of large-scale data centers has been deployed. In addition, cloud computing has been proposed recently to provide flexible and efficient services to the data service subscribers (DSSs). 
In cloud computing, the data service operator (DSO) is able to organize a shared pool of configurable computing resources (such as servers, storage, networks, applications and services), which can be easily accessed by DSSs on demands.

Generally speaking, large-scale data centers or clouds are generally built in remote areas which are far from the DSSs. This results in high transmission cost and service latency which can be intolerable for the IoT applications that requires real-time interaction or mobility. In order to overcome these challenges, fog computing is proposed as a promising solution. In fog computing, multiple low-power computing devices, commonly referred to as the fog nodes (FN), at the edge of the networks are deployed to offload the data computing services from the cloud. 
With the property of small-scale, low construction cost and mobility support, the FNs are generally deployed much closer to the DSSs and therefore can provide low-latency fast-response and location-awareness service\cite{IGoiri01}.

\begin{figure}
\centering
\includegraphics[scale=0.45, bb=577 369 45 726]{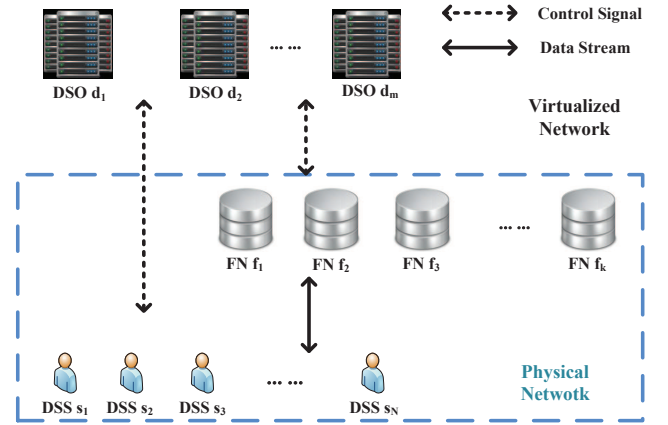}
\caption{System architecture}
\label{fig:architecture}
\end{figure}

Fog computing networks can consist of a large number of FNs deployed by different DSOs at different locations to provide various data services and applications to the DSSs. When DSSs can choose their DSOs as well as the corresponding FNs to further enhance their quality-of-experience, how to allocate the limited computing resources of all the fog nodes to the DSSs is still an open problem. In this paper, we further extend our previous work \cite{HZhang03} and focus on the resource selection and allocation problem between the FNs, DSOs and DSSs. We propose a joint optimization framework for all FNs, DSOs and DSSs in a distributed fashion. In the framework, we first formulate a Stackelberg game to model the interaction between DSOs and DSSs, where the DSOs set their service price first, and the DSSs purchase the optimal amount of computing resource blocks (CRBs) correspondingly. Once the prices of DSOs and the purchased resources of DSSs have been obtained, 
each DSS can know how many CRBs are required and can then try to compete for the CRBs owned by the nearby FNs. Thus, we propose a many-to-many matching game to investigate the interaction between DSOs and FNs where each DSO has a set of CRBs to offload, and each FN has many vacant CRBs to sell. After all the DSOs decide their DSO-FN pairs, FNs  will compete with each other to allocate their CRBs to the DSSs of the DSOs. We also adopt another layer of many-to-many matching framework to solve the FN-DSS pairing problem within the same DSO. Simulation results show that our proposed framework can significantly improve the performance of the fog computing networks. 

The rest of this paper is organized as follows. We describe the system model in Section \ref{sec:system model} and formulate the problems in Section \ref{sec:problem formulation}. Based on the formulated problem, we analyze the system with the proposed framework in Section \ref{sec:game analysis}, where the interaction between DSOs and DSSs is considered in Section \ref{sec:game analysis_DSO_vs_DSS}, the interaction between FNs and DSOs is analyzed in Section \ref{sec:game analysis_FN_vs_DSO}, and the interaction between FNs and DSSs is discussed in Section \ref{sec:game analysis_FN_vs_DSS}. Finally, we evaluate the performance of our work in Section \ref{sec:simulations}, show related works in Section \ref{sec:relatedwork} and summarize the paper in Section \ref{sec:conclusion}.

\section{System Model}\label{sec:system model}

\begin{table}

\caption{List of Notations}
\begin{tabular}{c|p{5.8 cm}}
\hline
Symbol & Definition\\
\hline
$M$ & Total number of DSOs \\
$N$ & Total number of DSSs \\
$K$ & Total number of FNs \\
$\Psi$ & The set of DSOs \\
$\Upsilon$ & The set of DSSs \\
$\Omega$ & The set of FNs \\
$\mu$ & Service rate of CRBs \\
$\lambda_{j}$ & Workload arrival rate for the DSS $s_j$ \\
$r_i$ & Price of unit virtualized CRB set by the DCO $o_i$ \\
$\mathbf{L}^s_j$ & The preference list of DSSs over DSOs \\
${t_{j}}$ & Total cost due to the delay of DSS $s_j$ \\
${h_{j}}$ & Cost due to network delay from the physical service provider to DSS $s_j$\\
$o_j$ & Cost due to queuing delay at the servers \\
$q_j$ & Total amount of CRBs purchased by the DSS $s_j$\\
$l_{kj}$ & Distance between the FN $f_k$ and the DSS $s_j$\\
$W^s_{j}$ & Utility function of DSS $s_j$ \\
$W^d_{i}$ & Utility function of DSO $d_i$ \\
$W^f_{k}$ & Utility function of FN $f_k$ \\
$\tau _{ij}$ & The boolean variable determining whether the DSO $d_i$ serves DSS $s_j$ or not.\\
$\alpha_j, \beta_j, \gamma_j$ & Weight factors in the utility function of DSS $s_j$\\
$t_{th}$ & The maximum tolerance of service delay for DSS $s_j$\\
$c_{kj}$ & Transmission cost for unit CRB from FN $f_k$ to DSS $s_j$\\
$e_{i}$ & Increment of the energy cost in the massive data center for DSO $d_i$\\
$\eta^{f}_{ki}$  & Normalized preference from the FN $f_k$ to the DSO $d_i$\\
$p_k$ & Rent of unit CRB set by the FN $f_k$\\
$r_{th}$ & Upper bound of total delay cost\\
$q_{kj}^{fs}$ &  Amount of CRBs allocated from the FN $f_k$ to the DSS $s_j$\\
$q_{ik}^{fd}$ &  Amount of CRBs allocated from the FN $f_k$ to the DSO $d_i$ \\
$\mathbf{L}^{df}_i$  & Preference list of DSO $d_i$ on all FNs \\
$\mathbf{L}^{sf}_j$  & Preference list of DSS $s_j$ on all FNs \\
$\mathbf{L}^{fs}_k$  & Preference list of FN $f_k$ on all DSSs \\
\hline
\end{tabular}\\

\end{table}

Consider a fog computing network where each DSS can submit its data computing service to a set of neighboring FNs deployed by a set of DSOs as illustrated in Figure \ref{fig:architecture}. Accordingly, we consider a three-tier fog network, where the DSOs locate in the middle layer, managing the FNs in the upper layer and serving DSSs in the bottom layer. Without loss of generality, we assume there are $M$ DSOs, labeled as $\Psi = \{d_1, d_2, \ldots, d_M\}$ and $N$ DSSs, denoted as $\Upsilon = \{s_1, s_2, \ldots, s_N\}$. 
Let ${\lambda _{j}}$ be the workload arrival rate of DSS $s_j$, $\forall s_j \in \Upsilon$. We assume each DSS has a normalized preference list, denoted as $\mathbf{L}^s_j$ over all DSOs. Moreover, $K$ FNs, labeled as $\Omega=\{f_1,f_2,\ldots,f_K\}$, locate in the area of consideration. We define the unit amount of computing resources that can be distributed by each FN as the ``computing resource block (CRB)'' \cite{HZhang03}, each of which can provide computing service at the rate of $\mu$. The physical data transmission network between FNs and DSSs satisfies the SecondNet topology \cite{CGuo01}, where the network facilities can provide the guaranteed quality-of-service (QoS) for the DSSs. Accordingly, in order to reduce the risk of potential network congestion and achieve real-time fast-response interaction, each DSO tries to offload the data services submitted by the DSSs to the large-scale data centers to the local FNs. However, as the DSSs cannot have the authorization to access the CRBs directly, the DSSs are required to receive the virtualized services from the DSOs, and with the management of DSOs, the CRBs of the FNs can finally be allocated to the DSSs. We assume that different DSOs offer data services with different requirements. Based on the preference list $\mathbf{L}^s_j$, the DSS $s_j$, $\forall s_j \in \Upsilon$ is required to subscribe to at most one DSO. The network architecture is illustrated in Fig. \ref{fig:architecture}.

Assume the DSSs apply real-time interactive applications, where QoS is measured by the service delay. In this paper, for DSS $s_j$, $\forall s_j \in \Upsilon$, we measure the cost of the service delay as

\begin{equation}
      {t_{j}} = {h_{j}}+ o_{j},
\end{equation}
which consists of the cost caused by the queuing delay $o_j$ at the servers as well as the cost caused by the network delay $h_j$ from the sensors to the physical service provider and from the physical service provider to DSS $s_j$.

According to the queuing delay model in \cite{AGan01, ZLiu01}, which can be easily extended to other models, the cost of queuing delay when serving DSS $s_j$ is

\begin{equation}\label{fun:delay}
      {o_{j}} = \frac{{{\lambda _{j}}}}{{{\mu} - \frac{{{\lambda _{j}}}}{{{q_{j}}}}}},
\end{equation}
where $q_{j}$ is the amount of CRBs purchased by DSS $s_j$. Moreover, as the network delay is related to the transmission distance, traffic condition in the network and many other unpredicted factors, in real situations, we suppose the network delay can be evaluated from training data periodically sent from sensors to the physical service provider and from the physical service provider to the DSS. In this paper, we set the distance between the farthest sensor to the physical service provider plus the distance between the physical service provider (e.g., FN $f_k$) to DSS $s_j$ is denoted as $l_{kj}$. For simplicity, we assume the network resource from the physical service provider to DSS $s_j$ is sufficient, and the cost incurred by the network delay $h_{j}$ generally follows a linear function of the distance from the sensor to the physical service provider plus the distance from the physical service provider to the DSS $s_j$, i.e., $h_{j}=\theta l_{kj}$, where $\theta$ is the scalar.

As the DSSs in the network pay DSOs for the service, following the structure of \cite{HZhang03}, the utility of DSS $s_j$, $\forall j \in \Upsilon$, can be denoted as the revenue received from the workload data minus both the cost of service delay and payment to the DSOs as follows:
\begin{equation}\label{utility_dss}
  W^s_{j} = \sum\limits_{i = 1}^M {{\tau _{ij}}} \left( {\alpha_{j} \lambda_{j} - \beta_{j} q_{j} r_{i}- \gamma_{j} t_{j}} \right),
\end{equation}
where $r_{i}$ is the price set by DSS $d_i$ for each unit of the virtualized CRB. $\alpha_{j}$, $\beta_{j}$, and $\gamma_{j}$ are weight factors. $\tau _{ij}$ is the boolean variable determining whether DSO $d_i$ serves DSS $s_j$ or not. If $\tau _{ij}=1$, DSS $s_j$ is served by the DSO $d_i$. Otherwise, DSS $s_j$ prefers to be served by other DSOs. The value of $\tau _{ij}$ follows the preference list $\mathbf{L}^s_j$ of DSS $s_j$, and since each DSS can at most choose one DSS, i.e., $\sum\limits_{i = 1}^M {{\tau _{ij}}}  = 1$, $\forall j \in \Upsilon$. For each DSS, assume there is an upper bound $t_{th}$ for the service delay. When the service delay is larger than the threshold, the DSS will regard it as an unsuccessful connection. Corresponding, we set $q_{j}^{th}$ as the lower bound of CRBs required for DSS $s_j$ to guarantee the service delay within the acceptable boundaries.

Based on the amount of virtualized CRB purchased by serving DSSs, the utility of each DSO is the revenues received from the DSSs minus the payment to the facilities that are able to provide the physical CRBs. Each DSO prefers to offload its services to the FNs nearby. However, if there are not sufficient amount of CRBs from the FNs which can meet the requirements of all DSSs. Some DSSs will be served by the massive data centers, which are located far away from the DSSs. We suppose the increment of the energy cost in the massive data center is $e_{i}$. Therefore, for the DSO $d_i$, if $q^{fd}_{ik}$ amount of CRBs are offloaded to the FN $f_k$, and $q^{o}_{i}$ amount of CRBs are still served by the massive data centers, the utility function of DSO $d_i$, $\forall i \in \Psi$, can be denoted as

\begin{equation}\label{utility_dso}
  W^d_{i} = \sum\limits_{j = 1}^N {{\tau _{ij}}({r_i}{q_j}) - \sum\limits_{k = 1}^K {{p_k}q_{ik}^{fd}}  - {e_i}q_i^o}.
\end{equation}
$p_k$ is the price set by the FN $f_k$, which is determined by the cost and current traffic of FN $f_k$.

For FN $f_k$ in the network, the utility is the payment received from the DSOs minus the transmission cost. We set $c_{kj}$ as the transmission cost for each unit CRB, which has a linear relationship with the distance $l_{kj}$. Moreover, we set $\eta^{f}_{ki}$ as the normalized preference to the DSO $d_i$. Accordingly, considering the preference to different DSOs, the discounted utility function for each DSO is

\begin{equation}\label{utility_fn}
  W^f_{k} = \sum\limits_{k = 1}^K {\eta^{f}_{ki}\left( {p_k} -  c_{kj} \right) q_{kj}^{fs}  },
\end{equation}
where $q_{kj}^{fs}$ is the amount of CRBs allocated from FN $f_k$ to DSS $s_j$.

\begin{figure}
\centering
\includegraphics[scale=0.5, bb=457 270 30 726]{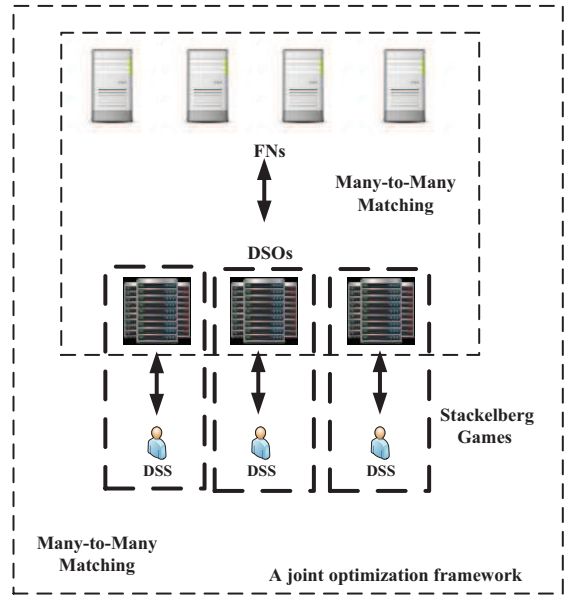}
\caption{A joint optimization framework}
\label{fig:game_model}
\end{figure}

\section{Problem Formulation}\label{sec:problem formulation}

According to the modeled architecture of fog network, with tradings between FNs and DSOs and between DSOs and DSSs, it is impossible to reach the maximum utilities for all FNs, MDCOs and DSSs simultaneously. Accordingly, we consider a sequential decision making process. During the process, the DSOs first predict the total amount of CRBs purchased from their servings DSSs and set their service prices to their DSSs based on the prediction. Observing the behaviors of the DSOs, each DSS determines the optimal amount of CRBs to purchase to achieve the maximum utility. Furthermore, knowing the total amount of required CRBs, the DSOs try to allocate CRBs to the FNs in the neighborhood. Finally, all FNs selected by the same DSO competes for the DSSs.

We can write the optimization problem can be formulated as follows.

\begin{equation}
    \begin{array}{l}
     \begin{gathered}
 \mathop {\max }\limits_{q_j} \;{\kern 1pt} \,W^s_{j} (  {q_j} | \mathbf{r} ), \qquad~~~~ \forall j \in \Upsilon,\hfill \\
  s.t. \left\{  {\begin{array}{*{20}{l}}
  {\sum\limits_{i = 1}^M {{\tau _{ij}}}  = 1}, \\
  {t_j \leq t_{th}}, \\
  {\tau _{ij} \alpha_{j} \lambda_{j} \leq \tau _{ij} \left( \beta_{j} q_{j} r_{i}+ \gamma_{j} t_{j}\right)},\\
  {q_j \geq 0}, \\
  {\tau _{ij} \in \{0,1\} , ~~~~\forall i\in \Psi},
\end{array}} \right.  \hfill
\end{gathered}
    \end{array}
    \label{opt:FN}
\end{equation}
where $\mathbf{r}$ is the pricing profile of all DSOs observed by DSS $s_j$.

Predicting the behaviors of the DSSs, each DSO is required to set the service price to each serving DSS and compete with other DSOs to choose FNs in the neighborhood. Thus, the formulated problem of the DSO is

\begin{equation}
    \begin{array}{l}
     \begin{gathered}
 \mathop {\max }\limits_{r_i} \;{\kern 1pt} \,W^d_{i} (  {r_i} | \mathbf{q}^*, \mathbf{p}, \mathbf{r}^*_{-i} ), \qquad~~~~ \forall i \in \Psi,\hfill \\
  s.t. \left\{  {\begin{array}{*{20}{l}}
  {\sum\limits_{j = 1}^N {{\tau _{ij}}({r_i}{q_j}) \geq \sum\limits_{k = 1}^K {{p_k}q_{ik}^{fd}}  + {e_i}q_i^o}}, \\
  {r_i \geq 0},
\end{array}} \right.  \hfill
\end{gathered}
    \end{array}
    \label{opt:FN}
\end{equation}
where $\mathbf{q}^*$ denotes the optimal amount of CRBs purchased by all DSSs. $\mathbf{p}$ is the profile of rent for all FNs. $\mathbf{r}^*_{-i}$ is the profile of optimal service prices set by other DSOs.

Moreover, each FN in the network would like to choose its preferred DSOs and serve its DSSs with low distance. Thus, competing with other FNs, it is required to determine the amount of CRBs allocated to DSOs and its serving DSSs, respectively. The optimization problem is denoted as

\begin{equation}
    \begin{array}{l}
     \begin{gathered}
 \mathop {\max }\limits_{q_{ik}^{fd}, q_{kj}^{fs}} \;{\kern 1pt} \,W^f_{k} (  q_{ik}^{fd}, q_{kj}^{fs} | \mathbf{q}^*, \mathbf{q}_{i{-k}}^{fd*}, \mathbf{q}_{{-k}j}^{fs*}, \mathbf{p}_{-k} ), \qquad~~~~ \forall k \in \Omega,\hfill \\
  s.t. \left\{  {\begin{array}{*{20}{l}}
  {\sum\limits_{i = 1}^M {q_{ik}^{fd}}  \leq q_k^{f-th}}, \\
  {\sum\limits_{j = 1}^N {q_{kj}^{fs}}  \leq q_k^{f-th}}, \\
  {\sum\limits_{k = 1}^K {q_{ik}^{fd}}  \leq q_i^{d-th},~~~~\forall i \in \Psi},  \\
  {\sum\limits_{k = 1}^K {q_{kj}^{fs}}  \leq q_j^{s-th},~~~~\forall j \in \Upsilon}, \\
\end{array}} \right.  \hfill
\end{gathered}
    \end{array}
    \label{opt:FN}
\end{equation}
where $\mathbf{q}_{i{-k}}^{fd}$ is the optimal amount of CRBs rent to DSO $d_i$ for all other FNs, $\mathbf{q}_{{-k}j}^{fs*}$ is the optimal amount of CRBs allocated to DSS $s_j$ for all other FNs. During the service, the total CRBs distributed to all DSOs or all DSSs cannot exceed its total available CRBs $q_k^{f-th}$ for FN $f_k$. Furthermore, the total CRBs purchased from DSO $d_i$ or DSS $s_j$ should not exceed its demand $q_i^{d-th}$ or $q_j^{s-th}$, respectively.

In summary, following the relationships among all FNs, DSOs and DSSs, we focus on the following problems:
\begin{enumerate}
  \item \textbf{Resource purchasing problem for the DSSs}: In the network, as the DSSs can only access to the DSOs in a virtualized fashion, they are required to purchase the optimal amount of CRBs from the DSOs. Following the system model, as different DSSs have different tolerance of service delay, when the upper bound of service delay is high, the DSSs are able to purchase a small amount of CRBs to achieve satisfying services. However, when the upper bound of service delay is low, the DSSs have to purchase a large amount of CRBs to guarantee the service delay is within the tolerated region. Moreover, the service price set by the DSOs also affects the utility of DSSs. When the price is in high value, even though the large amount of CRBs is able to improve the quality of data services, the DSSs have to make a large payment to the DSOs for their services. The revenue may not be satisfying. Therefore, considering both the delay tolerance and setting prices of DSOs, the optimal amount of CRBs should be determined for high utilities.
  \item \textbf{Pricing problem for the DSOs}: In the data service with fog, the DSOs are required to provide virtualized CRBs to the DSSs and try to rent the CRBs from the FNs to serve DSSs in the physical network. Therefore, how to do the pricing for the DSOs is a problem. Considering the announced rent from all FNs, the DSOs need to set a price which can bring profits for themselves. However, if the price is set too high, the serving DSSs will not purchase a large amount of CRBs. Therefore, predicting the reactions of CRBs and observing the rent of FNs, the DSOs are required to determine its service price so as to receive the maximum revenues.
  \item \textbf{DSO-FN pairing problem}: As FNs may be private computing devices, which are small-scale and unable to communicate with DSSs directly, the FNs are accessible to the DSOs only. In the multi-DSO scenario, as the FNs are accessible to all DSOs, it is a problem for all DSOs in the network to pair FNs distributedly so as to serve their DSSs with low latency. With different relations and trading history, each FN has different preference on all DSOs. Observing the rent of all FNs, each DSO also has a preferences on each FN. Based on the preferences of all DSOs' and FNs', it is required to reach a stable DSO-FN pairing results, where any DSO or FN is able switch its current pairing result for a higher utility.
  \item \textbf{FN-DSS pairing problem}: After the pairing between DSOs and FNs, each FN has allocated its CRBs to all DSOs. However, within one DSO, it is still a problem for the FNs to allocate their CRBs to all DSSs. As the distance between each FN and each DSS is various, with a longer transmission distance, the FNs have to pay more on the transmission cost. Thus, based on the transmission distance, each FN has a preference over all DSSs. Moreover, each DSS also have preference over FNs based on the rent. Therefore, following the preference of all FNs' and DSSs', a stable FN-DSS pairing result should be achieved.
\end{enumerate}

According to the formulated problems, all FNs, DSOs and DSSs are rational and autonomous individuals, who observe the behaviors of others and make decisions to improve their own utilities. Therefore, in order to reach the optimal and stable solutions for all FNs, DSOs and DSSs, we model a three-stage joint optimization framework, as shown in Fig. \ref{fig:game_model}. In the framework, we first model the Stackelberg games to solve the pricing problem for the DSOs and resource purchasing problem for the DSSs. When the DSOs know the expected amount of resource purchased by the DSSs, a many-to-many matching is proposed between the DSOs and the FNs to deal with the DSO-FN pairing problem. Finally, within the same DSO, we apply another many-to-many matching between its paired FNs and serving the DSSs to solve the FN-DSS pairing problem.

\section{System Analysis}\label{sec:game analysis}

In this section, we analyze the optimal strategies for FNs, DSOs and DSSs. 
Based on the analysis of the formulated framework, in the following sub-sections, we first investigate the interactions between the DSOs and DSSs to determine how many CRBs are required during the service. Given the optimized behaviors of the DSOs and DSSs, we analyze the interactions between the FNs and DSOs based on different preferences. Finally, with the obtained results, we discuss the interactions between the FNs and DSSs within the same DSO for better services.

\subsection{The Interaction between DSOs and DSSs}\label{sec:game analysis_DSO_vs_DSS}

In the virtualized network, the DSOs provides CRBs for the DSSs. Following the formulated problems for both DSOs and DSSs, there is a Stackelberg game, where the DSOs act as leaders and DSSs act as followers. In the game, when all DSSs choose their serving DSOs with their preferences, the DSO first sets the service price. Then, based on the price all DSSs determine optimal amount of CRBs to purchase. Accordingly, considering the optimization problem of DSSs, we have the following lemma.

\begin{lemma} \label{lemma1}
In the modeled Stackelberg game between DSO $d_i$ and DSS $s_j$, when the DSO announces its service price $r_{i}$, the optimal amount of CRBs $q_{j}$ purchased by the DSS is

\begin{equation}\label{one-to-one}
    q^{*}_{j}=\frac{{{\lambda _{j}}}}{{{\mu }\sqrt {r_{i}\frac{\beta_j}{\gamma_j}} }} + \frac{{{\lambda _{j}}}}{{{\mu}}}.
\end{equation}
\end{lemma}

\begin{proof}
According to the utility function of DSS $s_j$ \ref{utility_dss}, the second derivative of $W^s_{j}$ with respect to $q_{j}$ is

\begin{equation}
   \frac{\partial ^2{W^s_{j}}}{\partial q^{2}_{j}} =  - \frac{{2\lambda _{j}^2{\mu}}}{{{{({\mu}{q_{j}} - {\lambda _{j}})}^3}}}.
\end{equation}
As $\frac{\partial ^2{W^s_{j}}}{\partial q^{2}_{j}} < 0$, $W^s_{j}$ is a quasi-concave function with respect to $q_{j}$. Furthermore, the first derivative of $W^s_{j}$ with respect to $q_{j}$ is

\begin{equation}
   \frac{{\partial {W^s_{j}}}}{{\partial {q_{j}}}} = {\left( {\frac{{\lambda _{j}}}{{{\mu}{b_{j}} - {\lambda _{j}}}}} \right)^2} - {r_{i}}\frac{\beta_j}{\gamma_j}.
\end{equation}
We set the first derivative equal to zero and obtain the optimal amount of CRBs to purchase so as to achieve the maximum utility, i.e,

\begin{equation}
    q_{j}^*=\frac{{{\lambda _{j}}}}{{{\mu }\sqrt {r_{i}\frac{\beta_j}{\gamma_j}} }} + \frac{{{\lambda _{j}}}}{{{\mu}}}.
\end{equation}
\end{proof}

Therefore, considering the reactions of the DSSs, we adjust the optimization problem for the DSO $d_i$, $\forall i \in \Psi$, as follows.

\begin{equation}
    \begin{array}{l}
     \begin{gathered}
 \mathop {\max }\limits_{r_i} \;{\kern 1pt} \,\widetilde{W}^{d}_{i} (  {r_i} | \mathbf{q}^*, \mathbf{p}, \mathbf{r}^*_{-i} ), \qquad~~~~ \forall i \in \Psi,\hfill \\
  s.t. \left\{  {\begin{array}{*{20}{l}}
  {\sum\limits_{j = 1}^N {{\tau _{ij}}(\frac{{{\lambda _{j}}}}{\mu }\sqrt {\frac{\beta_j}{\gamma_j}}{r_i}  + \frac{{{\lambda _{j}}}}{\mu }\sqrt{\frac{\beta_j}{\gamma_j}}{r_i}) \geq \sum\limits_{k = 1}^K {{p_k}q_{ik}^{fd}}  + {e_i}q_i^o}}, \\
  {r_i \geq 0},
\end{array}}  \right.  \hfill
\end{gathered}
    \end{array}
\label{opt:DSO_2}
\end{equation}
where

\begin{equation}
    \widetilde{W}^{d}_{i} = \sum\limits_{j = 1}^N {{\tau _{ij}}(\frac{{{\lambda _{j}}}}{\mu }\sqrt {\frac{\beta_j}{\gamma_j}}{r_i}  + \frac{{{\lambda _{j}}}}{\mu }\sqrt{\frac{\beta_j}{\gamma_j}{r_i}}) - \sum\limits_{k = 1}^K {{p_k}q_{ik}^{fd}}  - {e_i}q_i^o}.
\end{equation}

In the formulated problem (\ref{opt:DSO_2}), we take the first derivative of $\widetilde{W}^{d}_{i}$ with respective to $r_i$ and discover it is a monotonous increasing function with respective to $r_i$. Furthermore, as the service delay cannot surpass $t_{th}$ for the DSSs. The CRB purchased by DSSs has the following low bound

\begin{equation}\label{threshold_b}
    q_{j}\geq \frac{{{\lambda _{j}}{t_{th}}}}{{\mu {t_{th}} - {\lambda _{j}}}}.
\end{equation}
Thus, following the relation in (\ref{one-to-one}), the maximum and optimal price for the DSO $d_i$ to DSS $s_j$ is, $\forall i \in \Psi$, $\forall j \in \Upsilon$,

\begin{equation}
    r_i=\frac{{{\gamma _j}}}{{{\beta _j}}}{\left( {\frac{{\mu {t_{th}} - \lambda_j }}{\lambda_j }} \right)^2}.
\end{equation}

\subsection{The Interaction between FNs and DSOs}\label{sec:game analysis_FN_vs_DSO}

According to the predicted amount of CRBs ordered by the DSSs, the DSOs try to offload the services from the massive data centers to the DSSs nearby. Observing the service prices set by all FNs, DSO $d_i$, $\forall i \in \Psi$, has a preference list $\mathbf{L}^{df}_i=\left[{L}^{df}_{i1},\ldots,{L}^{df}_{iK}\right]$ on all FNs in the neighborhood. As the DSOs prefer to choose the FNs with a low price, we set

\begin{equation} \label{pl:df_d}
    {L}^{df}_{ik} = -p_k, ~~~~~~~~~~~` \forall i \in \Psi,~\forall k \in \Omega.
\end{equation}
Furthermore, each FN also has different preferences over DSOs. Thus, we set the preference list of FN $f_k$, $\forall k \in \Omega$, on all DSOs as $\mathbf{L}^{fd}_k=\left[{L}^{fd}_{1k},\ldots,{L}^{df}_{Mk}\right]$, satisfying,

\begin{equation}\label{pl:df_f}
    {L}^{fd}_{ik} = \eta^{f}_{ki}, ~~~~~~~~~~~` \forall i \in \Psi,~\forall k \in \Omega.
\end{equation}

Considering the preference lists of FNs and DSOs, i.e., ${L}^{df}_{ik}$ and ${L}^{fd}_{ik}$, respectively, we design a many-to-many matching algorithm for the DSO-FN pairing problem. As shown in Algorithm \ref{algorithm1}, after the preference lists are constructed, we set a pointer for each FN in its preference list. Initially, we set the pointer at the first DSO in the list. During the each round of matching, each FN first propose to the DSO in the pointer of the preference list. Based on behaviors of the FNs, each DSO chooses its most preferred FNs in its preference list until the required CRBs are all supplied by the FNs. If the FNs supply more CRBs than the DSO requires, the DSO will reject the CRBs from less favourite FNs. If the FNs supplies less CRBs than the DSO requires, the DSO will choose massive data centers for the rest of the services. At the end of each round, if all of the CRBs from the FN have been allocated to the DSOs, the pointers of the FN will keep unchanged. Otherwise, the pointers of the rejected FNs will move to the next DSO in the preference list. In the next round, the FNs which still have available CRBs will propose to the new DSOs according to their pointers. Specifically, if the CRBs of FN are chosen by the DSO in the last round, but discarded in the current round, we suppose the pointer of the FN doesn't change its position, considering the pointed DSO in the current round may need more CRBs from the FN. The matching repeats in circulations until all the pointers of the FNs have moved out their preference list. According to the algorithm, we have the following lemmas.

\begin{algorithm}[t]
\caption{Many-to-Many Matching Algorithm for DSO-FN pairing problem.}
\vspace{.1cm}
\hrule
\begin{algorithmic}[1]
\vspace{.2cm}
        \FOR{FN $f_k$}
            \STATE Construct the preference list $\mathbf{L}^{fd}_{k}$ on all DSOs according to (\ref{pl:df_f});
            \STATE One pointer is set as the indicator pointing at the first DSO in the preference list.
        \ENDFOR
        \FOR{DSO $d_i$}
            \STATE Construct the preference list $\mathbf{L}^{df}_{i}$ on all FNs according to (\ref{pl:df_d});
        \ENDFOR
        \STATE We set $flag_k$, $\forall k \in \Omega$, as the indicator to show if the CRBs of FN $f_k$ were chosen by the DSO in the last round, but discarded in the current round. Initially, $flag_k=1$;
        \WHILE{the pointers of all FNs have not scanned all the DSOs in their preference list}
            \STATE FNs propose to DSOs with their service price;
            \FOR {FN $f_k$ who still have available FNs to purchase}
                \IF {$flag_k=1$}
                    \STATE The pointer keep current position in the list;
                \ELSE
                    \STATE The pointer moves to the next position in the list;
                \ENDIF
                \STATE The FN proposes to pointed DSO in its preference list with its available FNs;
                \STATE We set $flag_k=0$;
            \ENDFOR
            \STATE DSOs determines which FNs to choose;
            \FOR{DSO $d_i$}
                \IF {The total available amount of CRBs proposed by the FNs exceed the requirements }
                    \STATE The DSO $d_i$ chooses the most preferred amount of CRBs from FNs, and rejects the rest;
                    \STATE For CRBs of the FN $f_k$ which is chosen by the DSO in the last round, but rejected in the current round, we set $flag_k=1$;
                \ENDIF
            \ENDFOR
         \ENDWHILE
\end{algorithmic}\label{algorithm1}
\hrule
\end{algorithm}

\begin{lemma} \label{lemma2}
For each FN in the matching algorithm, the pointer of the FN in its preference list moves in one direction. In other words, when its pointer has moved to the next DSOs in the preference list, whatever the matching results of other FNs, the FN cannot achieve a higher utility by moving the pointer back.
\end{lemma}

\begin{proof}
As shown in Algorithm \ref{algorithm2}, when the DSOs determines which FNs to choose, they choose the CRBs from the most preferred amount. If some CRBs from FN $f_k$ is discarded by DSO $d_i$, the current accepted CRBs belong to the FNs which is more preferred than FN $f_k$. In the future rounds, when there are new FNs proposing to DSO $d_i$, if the DSO would like to change its current accepted FNs, the DSO can only choose the FNs that is even better than the FNs in the current accepted list. Therefore, for FN $f_k$ which has been rejected once from DSO $d_i$, it is impossible for it to be accepted by the same DSO in the future rounds. \end{proof}

\begin{lemma} \label{lemma3}
Following the Algorithm \ref{algorithm2}, the DSO-FN pairing problem will ultimately converge and achieve a stable matching result.
\end{lemma}

\begin{proof}
As proved in the Lemma \ref{lemma2}, the pointer of the FNs can only move in one direction. Therefore, in the perspective of the FNs, when the pointer of each FN has moved to the end of the preference list, the FN has distributed its available CRBs to the DSOs in an optimized way. In other words, the FN is unable to change its pairing results unilaterally for higher utilities. Furthermore, in the perspective of the DSOs, when the pointers of all FNs have moved to the end of the preference lists, each DSO has evaluated all proposals from the available FNs. Therefore, it also cannot unilaterally change its pairing results, get accepted by the FNs and achieve a higher utility for itself. Furthermore, according to \cite{Roth01}, when every agent¡¯s preference list is substitutable, a pairwise stable matching always exists. Based on the above, the DSO-FN pairing will ultimately converge and achieve a stable matching result in Algorithm \ref{algorithm2}. \end{proof}

\subsection{The Interaction between FNs and DSSs}\label{sec:game analysis_FN_vs_DSS}

When the CRBs from FNs have been rent to all DSOs, within each DSO, how to allocate the CRBs to all DSSs still remains a problem. According to the rent of all FNs, DSS $s_j$, $\forall j \in  \Upsilon$, has a preference list $\mathbf{L}^{sf}_j=\left[{L}^{sf}_{j1},\ldots,{L}^{sf}_{jK}\right]$ on all FNs in the neighborhood, satisfying,

\begin{equation} \label{pl:fs_s}
    {L}^{sf}_{kj} = -p_k, ~~~~~~~~~~~~ \forall j \in \Upsilon,~\forall k \in \Omega.
\end{equation}

Furthermore, according to the utility function of the FN, the distance between the FN and its serving DSS affect the revenues the FN received. With a longer distance, the FN has to pay more for the data transmission in the network. Therefore, we set a preference list $\mathbf{L}^{fs}_i=\left[{L}^{sf}_{j1},\ldots,{L}^{sf}_{jK}\right]$ for FN $f_k$, $\forall k \in \Omega$, over all DSSs, i.e.,

\begin{equation}\label{pl:fs_f}
    {L}^{fs}_{kj} = -l_{kj}, ~~~~~~~~~~~~ \forall k \in \Omega,~\forall j \in \Upsilon.
\end{equation}

Considering the preference lists of DSSs and FNs, i.e., ${L}^{sf}_{kj}$ and ${L}^{fs}_{kj}$, respectively, we design a many-to-many matching algorithm for the FN-DSS pairing problem within DSO $d_i$, $\forall i \in \Psi$. As shown in Algorithm $\ref{algorithm2}$, after the preference lists are constructed, we set a pointer for each DSS in its preference list. Initially, we set the pointer at the first FN in the list. During the each round of matching, each DSS first proposes to the FN in the pointer of the preference list. Based on behaviors of the DSSs, each FN chooses its most preferred DSSs in its preference list until the maximum CRBs available in the DSO $d_i$ are reached. If the DSSs request more CRBs than the FN can supply, the FN will reject the less favourite DSSs. At the end of each round, if the DSSs have been allocated all of its requested CRBs from the FNs, the pointers of the DSS will keep unchanged. Otherwise, the pointers of the rejected DSS will move to the next FN in the preference list. In the next round, the DSSs which require CRBs will propose to the new FNs according to their pointers. Specifically, if some CRBs of FN are allocated to the DSS in the last round, but changed to other DSSs in the current round, we suppose the pointer of the DSS in the next round doesn't change its position, considering the pointed FN in the current round may be able to supply more CRBs to the DSS. The matching repeats in circulations until all the pointers of the DSSs have moved out their preference list. Following Lemma \ref{lemma2} and Lemma \ref{lemma3} in a similar way, the FN-DSS pairing problem can ultimately achieve a stable matching result.

\begin{algorithm}
\caption{Many-to-Many Matching Algorithm for FN-DSS pairing problem.}
\vspace{.1cm}
\hrule
\begin{algorithmic}[1]
\vspace{.2cm}
        \FOR{DSS $s_j$}
            \STATE Construct the preference list $\mathbf{L}^{sf}_{j}$ on all DSOs according to (\ref{pl:fs_s});
            \STATE One pointer is set as the indicator pointing at the first FN in the preference list.
        \ENDFOR
        \FOR{FN $f_k$}
            \STATE Construct the preference list $\mathbf{L}^{fs}_{k}$ on all FNs according to (\ref{pl:fs_f});
        \ENDFOR
        \STATE We set $flag_j$, $\forall j \in \Upsilon$, as the indicator to show if the FNs allocate CRBs to the DSS $s_j$ in the last round, but adjusted in the current round. Initially, $flag_j=1$;
        \WHILE{the pointers of all DSSs have not scanned all the FNs in their preference list}
            \STATE DSSs propose to FNs for their services;
            \FOR {DSS $s_j$ which has not been allocated required CRBs}
                \IF {$flag_j=1$}
                    \STATE The pointer keep current position in the list;
                \ELSE
                    \STATE The pointer moves to the next position in the list;
                \ENDIF
                \STATE The DSS proposes to pointed FN in its preference list for its data services;
                \STATE We set $flag_j=0$;
            \ENDFOR
            \STATE Each FN determines which DSSs to choose;
            \FOR{FN $f_k$}
                \IF {The total available CRBs requested by the DSSs exceed the available volume }
                    \STATE The FN $f_k$ allocate the CRBs to the most preferred DSSs, and rejects the rest;
                    \STATE For CRBs allocated to the DSS $s_j$ in the last round, but adjusted in the current round, we set $flag_j=1$;
                \ENDIF
            \ENDFOR
         \ENDWHILE
\end{algorithmic}\label{algorithm2}
\hrule
\end{algorithm}

\section{Simulation Results and Discussions}\label{sec:simulations}

In this section, we present simulation results to evaluate our proposed framework with MATLAB. In the simulated IoT scenario, without specific explanation, there are 120 DSSs, 4 DSOs and 20 FNs allocated randomly in a circle district with a diameter of 10 kilometers. As we focus on the IoT scenarios where DSSs are closely located with its sensors, without loss of generality for the methods in this paper, we suppose the sensors of each DSS are located at the same position with the DSS. We follow the settings in \cite{ZLiu01} and set the service rate of each computing resource block is $0.1$ $(ms)^{-1}$. For each DSS, the workload arriving rate is a random number averaged $0.5$ $(ms)^{-1}$. The data transmission speed is $50 km/ms$. The delay tolerance of all DSSs is set to $60$ ms. Furthermore, for each FN, we set its preference to each DSO as a random number satisfying the uniform distribution between $\left[0,1\right]$. Based on the usage of its computing resources and its service cost, we set the announced rent as a random number satisfying the uniform distribution between $\left(0,10\right)$, and the amount of available CRB as a random number satisfying uniform distribution between $\left(0,100\right)$. The weight factors $\alpha$, $\beta$ and $\gamma$ are set as $50$, $0.01$ and $0.001$, respectively.

As shown in Fig. \ref{fig:numSS_utilityFN}, we evaluate the utility of all FNs when the number of DSSs increases. When the number of FNs is fixed, we discover with the number of DSSs increasing, the utility of all FNs generally increases, and the increasing speed first increases then gradually decreases to zero. The reason is that when the number of DSSs increases, but the number of FNs is fixed, the FNs will be able to serve more favourable DSSs with a low transmission distance. However, when all of the available CRBs of FNs are allocated to the DSSs nearby and the number of DSS keep increasing, the DSS will be allocated with CRBs from the massive data centers. Thus, the total utility of the FNs stop increasing. Nevertheless, when we increase the number of FNs, we discover that with more FNs, the utility can converge to a higher bound ultimately. Furthermore, because of the competition between FNs, when the number of DSS is small, with more FNs, the increasing speed is smaller.

\begin{figure}[t]
\centering
\includegraphics[scale=0.55, bb=389 0 0 303]{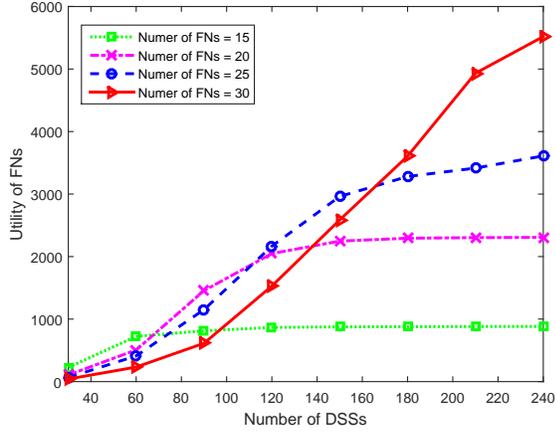}
\caption{The utility of all FNs versus the number of DSSs}
\label{fig:numSS_utilityFN}
\end{figure}

In Fig. \ref{fig:numSS_utilitySS}, we consider the utility of all DSSs with the number of DSSs increasing. In the simulation, we compare the performance of the DSSs in our proposed framework with the performance of the DSSs which is served by the massive data centers only. With the same amount of workload, we discover that when the number of DSSs increases, the utility of DSS genrally increases, and the utility with FNs achieve a higher value than the one without FNs. Furthermore, when the number of DSSs is fixed and the average workload for each DSS increases, the DSSs are able to receive more revenues from the services. Thus, the utility of DSSs increases.

\begin{figure}[t]
\centering
\includegraphics[scale=0.55, bb=389 0 0 303]{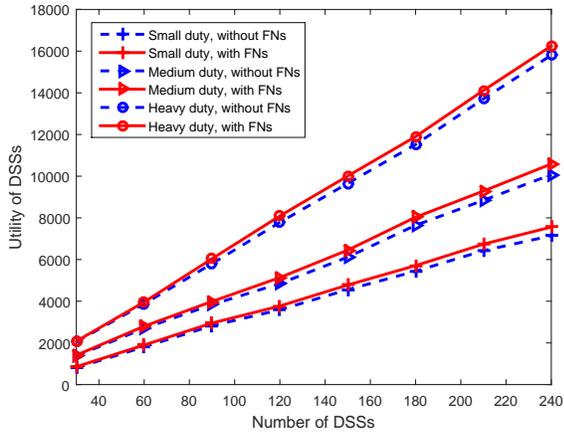}
\caption{The utility of all DSSs versus the number of DSSs}
\label{fig:numSS_utilitySS}
\end{figure}

In Fig. \ref{fig:lambda_utilityFN}, we analyze the relation between the utility of FNs and average workload arrive rate for DSSs. As shown in the figure, when the number of FNs is fixed and the average value of workload $\lambda$ increases, the utility of FNs first increases then gradually converge to a fixed value. The reason is that when the workload of all DSSs increases, the FNs are able to allocate more of its CRBs to the DSSs nearby. However, when all the available CRBs of FNs are allocated to the DSSs, the utility of the FNs stops improving and converges to one specific value. When the number of FNs increases, with the same value of $\lambda$, as the FNs are able to provide more CRBs to the DSSs, the converged value is higher.

\begin{figure}[t]
\centering
\includegraphics[scale=0.55, bb=389 0 0 303]{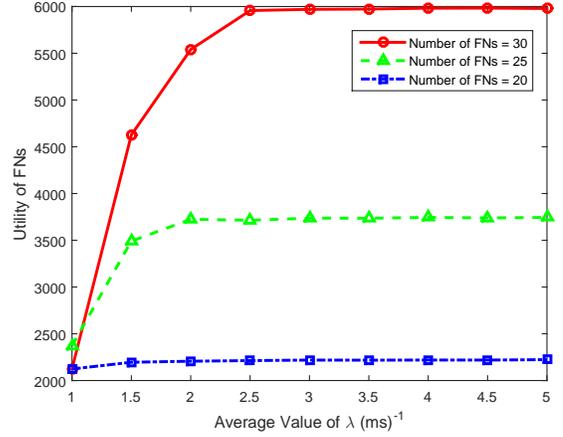}
\caption{The utility of all FNs versus the average value of $\lambda$}
\label{fig:lambda_utilityFN}
\end{figure}

In Fig. \ref{fig:utilityDSO_mu}, we observe the utility of all DSOs when the value of $\mu$ increases. As shown in the figure, when $\mu$ increases, each DSS is able to be served with a less amount of CRBs. Therefore, the DSO is able to set a higher price to the DSSs and receive high revenues. Moreover, when the number of DSSs increases, as the DSO is able to serve more DSSs at the same time, the total utility of DSOs also increases.

\begin{figure}[t]
\centering
\includegraphics[scale=0.55, bb=389 0 0 303]{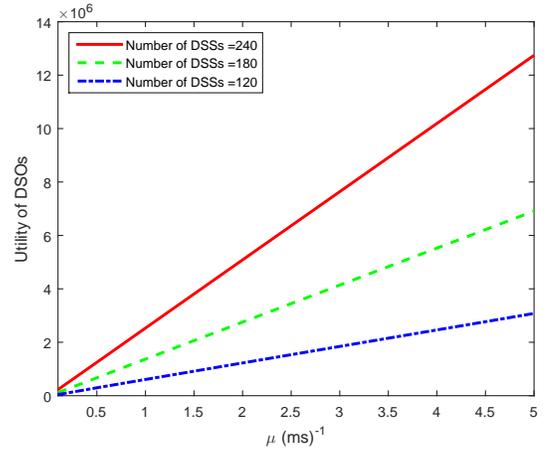}
\caption{The utility of all DCOs versus the value of $\mu$}
\label{fig:utilityDSO_mu}
\end{figure}

In Fig. \ref{fig:utility_rth}, we evaluate the relationship of utilities and the tolerance service delay $t_{th}$ of DSSs. As shown in the Fig. \ref{fig:utilityFN_rth}, when the value of $t_{th}$ increases, each DSS is able to be served with less CRBs. Thus, the FNs will supply less CRBs in the network, and the utility of the FNs generally decreases. In Fig. \ref{fig:utilityFN_rth}, as the DSO is able set a high price for its services, with the value of $t_{th}$ increasing, the utility of DSO generally increases. Moreover, in Fig. \ref{fig:utilityDSS_rth}, even though the DSS is able purchase less CRBs with high $t_{th}$, the price of CRBs set by the DSOs also increases. Furthermore, as the DSS suffers a lot with high delay, the utility of DSS generally decreases with the value of $t_{th}$ increasing.

\begin{figure*}[htbp]
\begin{minipage}[t]{0.33\linewidth}
\centering
\includegraphics[scale=0.42, bb=389 0 0 303]{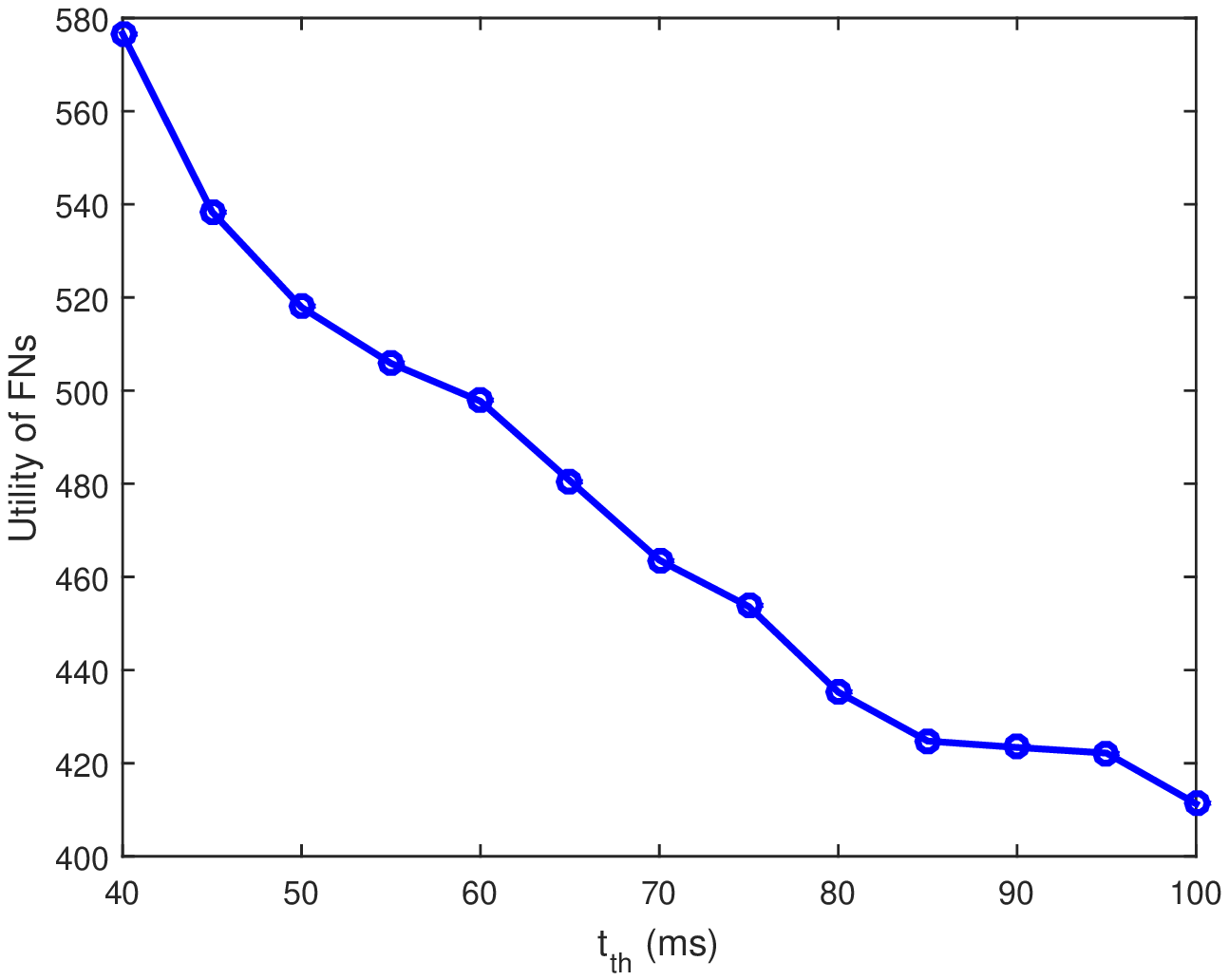}
\subcaption{The utility of all FNs versus $t_{th}$ }
\label{fig:utilityFN_rth}
\end{minipage}%
\begin{minipage}[t]{0.33\linewidth}
\centering
\includegraphics[scale=0.42, bb=389 0 0 303]{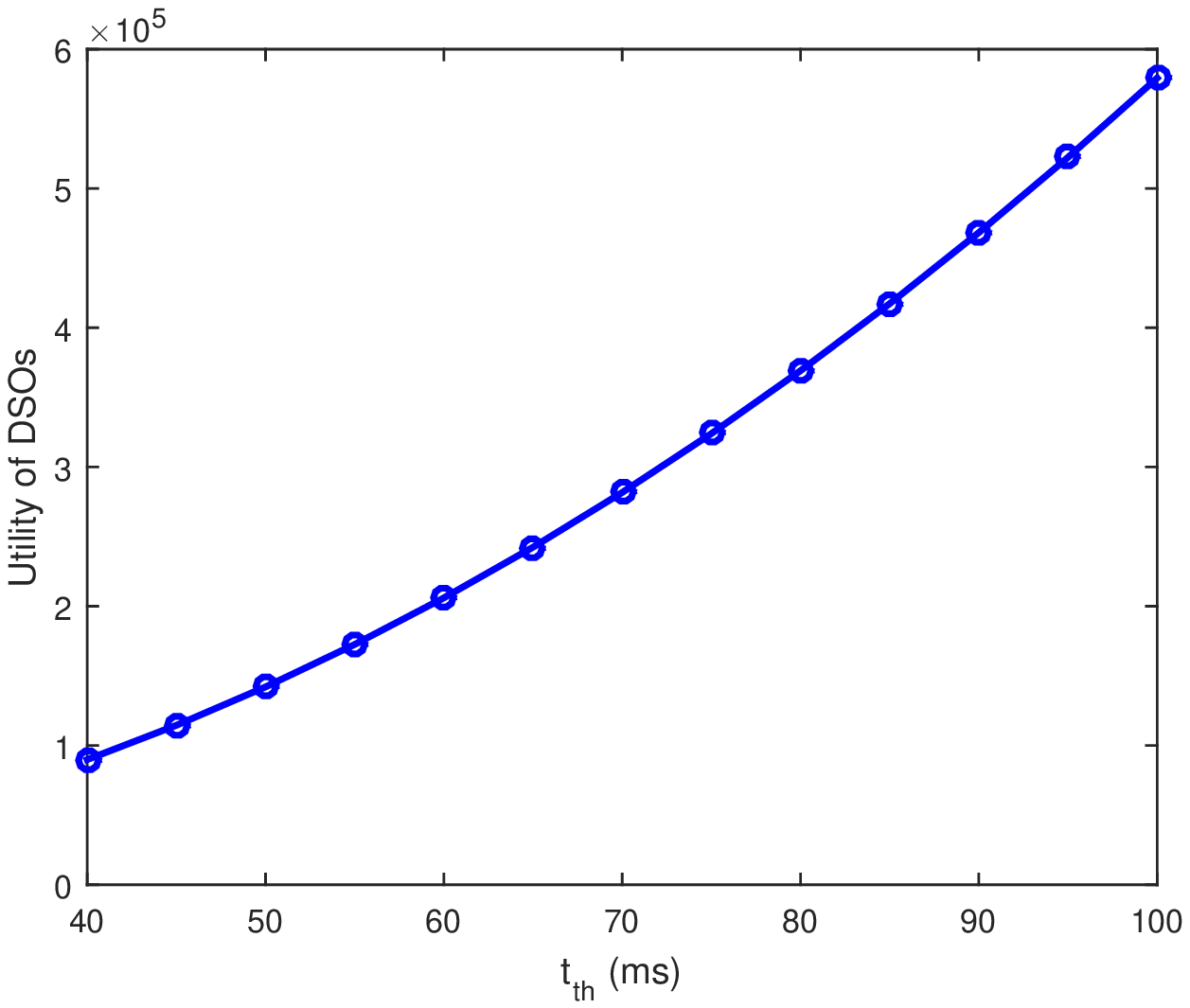}
\subcaption{The utility of all DSOs versus $t_{th}$ }
\label{fig:utilityDSO_rth}
\end{minipage}
\begin{minipage}[t]{0.33\linewidth}
\centering
\includegraphics[scale=0.42, bb=389 0 0 303]{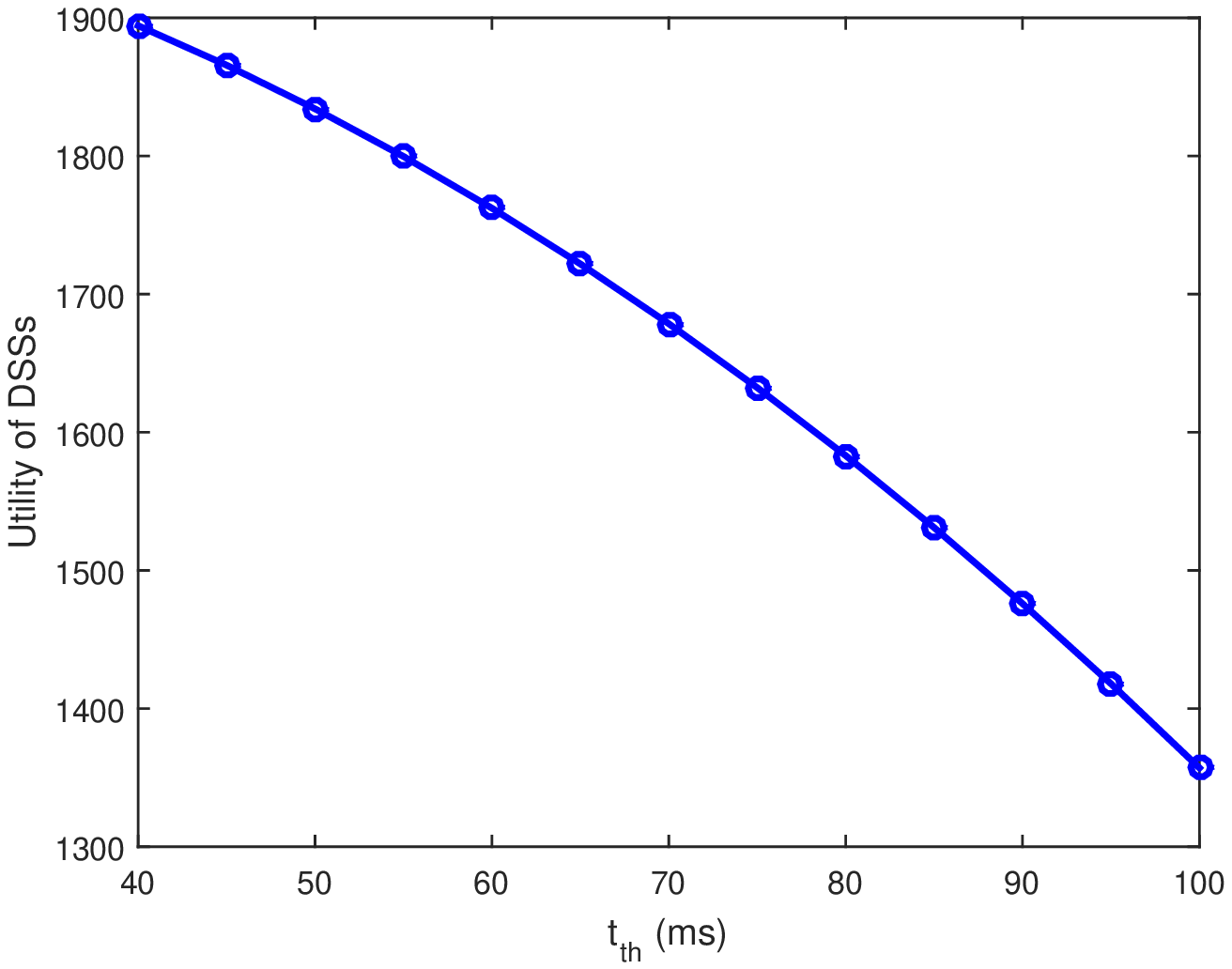}
\subcaption{The utility of all DSSs versus $t_{th}$ }
\label{fig:utilityDSS_rth}
\end{minipage}
\caption{The utilities versus $t_{th}$ }
\label{fig:utility_rth}
\end{figure*}

\section{Related Works}\label{sec:relatedwork}

In the literature, fog computing has been advocated to be the promising future of the cloud. The concept of pulling the cloud closer to the users has been widely considered in previous work. In \cite{BAhlgren01}, the authors put forward the concept of mist computing, aiming to distribute the cloud and its benefits deeply into the network. In \cite{SIslam01}, the deployment of edge cloud was proposed. From the DSSs' side, the edge cloud was able to surrogate the requirements and simplify the management of the network. From the servers' side, the edge cloud can exploit content and support service delivery in an efficient way. Without deploying massive data centers with high cost and latency, \cite{KChurch01} on the other hand strengthened the importance of small distributed data center designs. The authors took email delivery as an example and showed the advantages of geo-diversity characteristics of micro data centers. In \cite{VValancius01}, a novel and distributed computing platform, called nano data centers was proposed. The authors evaluated the energy consumption of nano data centers and showed a significant improvement on energy efficiency, compared with the traditional data centers. In \cite{MBMobley01}, considering the existing telecommunication and Internet service providers, the authors showed that it was required and beneficial to leverage the existing infrastructure and provide value-added services with FNs. In \cite{FBonomi01}, the authors outlined the vision of fog computing and overviewed the important features of fog computing. In \cite{Zjiang01}, the network optimization with fog computing was considered. As the data centers were aware of the locations of DSSs with fog computing, dynamic adaptation of computing resources to the DSSs' conditions was proposed. In \cite{Firdhous01}, the authors compared the cloud computing with the fog computing and showed some significant characteristics of fog, which was required for current data services. In \cite{Stojmenovic01}, the authors elaborated the role of fog computing in six important scenarios and surveyed the security issues with fog computing.

Moreover, fog computing has been widely considered to be beneficial for the IoT. In \cite{MChiang001}, the authors overview the opportunities and challenges of fog, especially the applications of fog computing in IoT. In \cite{MAazam001}, the authors devise the method of MEdia FOg Resource Estimation (MeFoRE), to provide resource estimation based on the service give-up ratio and to enhance QoS based on the previous QoE. \cite{SFAbedin001} addresses the utility based pairing problem between the fog nodes and IoT devices with the Irving's matching algorithm. In \cite{NKGiang001}, the authors propose a distributed dataflow programming model for IoT devices to optimize resource allocation on computing infrastructures across the fog and the cloud. In \cite{MAazam002}, the authors consider issues of resource prediction, customer type based resource estimation and reservation, advance reservation, and pricing in the fog computing for IoTs. \cite{MYannuzzi001} considers the requirements of mobility, scalability, reliable control and actuation in some challenging scenarios of IoT to show the benefits and significance of fog computing. Considering the advantages of fog computing, the authors in \cite{BMei001} discuss and propose a procedure to be implemented in smart phones for UV measurement.

In order to solve the resource management problems in a network system with a distributed fashion, game theory has been shown as a powerful tool \cite{ZHan001}. In the literature, most of the cases, the network system is normally modeled as a bipartite or a multi-tier graph. Based on this model, in \cite{CWang01}, a Stackelberg game theoretic model was shown for dynamic bandwidth allocation between virtual networks. In \cite{MWu001}, the authors consider a Stackelberg game between data center and buses in the smart city, where each buses collect data along its route and compete with other buses for the reward forwarding to the data center. In the game, following the proposed heuristic algorithm, the Stackelberg equilibrium is shown to be achieved where the data center and each bus are able to reach maximum utility. In \cite{YWang001}, the authors formulate a Stackelberg game for power allocation of data centers in the cloud. In the game, the power grid controller acts as the leader and sets prices of the provided energy based on the current amount of renewable energy and costs. Observing the prices, the cloud controller, i.e., the follower, determines the optimal amount of enery to purchase and do resource allocation for its data centers. With backward induction, the near-optimal strategies of both players in the game can be achieved. In \cite{HWang001}, the authors model the interaction between the monopolistic data center opertor and the customers as a Stackelberg game. In the game, the pricing strategies of the monopolistic data center operator the corresponding behavior of data service customers is detailedly analyzed in both homogeneous and heterogeneous customer scenarios. In \cite{BYang01}, the authors adopted the Stackelberg game to solve the problem of minimizing energy consumption in the data center networks. In \cite{HZhang01}, the authors proposed a multi-leader multi-follower Stackelberg game to address and cooperation problems among Wi-Fi, small cell and macrocell networks. In \cite{HZhang02}, the authors combined the Stackelberg game and the bargaining to design a resource allocation problem in a multi-tier LTE unlicensed network. Furthermore, the auction mechanisms are also powerful tools to solve the problem. In \cite{LMAusubel02, GDemange01, DMishra01}, the resource management problem could be perfectly optimized, but it requires high communication and computation overhead. In \cite{RNarayanam01}, the authors adopted the generic game theoretic framework to identify important edges in the context of k-edge connectivity between certain pairs of nodes in a general given network. In \cite{MKearns01}, the graphical game was put forward to analyze the optimized behaviors of each node in a general graph.

\section{Conclusions}\label{sec:conclusion}

In this paper, we proposed a joint optimization framework in the multi-FN, multi-DSO and multi-DSS scenario for IoT fog computing. In the framework, we first modeled the Stackelberg games to solve the pricing problem of the DSOs and resource purchasing problem of the DSSs. Then a many-to-many matching was proposed between the DSOs and the FNs to deal with the DSO-FN pairing problem. Finally, we applied another many-to-many matching between its paired FNs and serving DSSs to solve the FN-DSS pairing problem within the same DSO. For each stage of the problem, all participants were able to achieve the equilibrium or stable results where no one was able to change its behavior unilaterally for a higher utility. Simulation results showed that all FNs, DSOs and DCOs were able to reach optimal utilities for themselves, and high performance of the proposed framework could be achieved compared with the data services without fog nodes.


\begin{thebibliography}{1}

\bibitem{JGubbi001}
J. Gubbi, R. Buyya, S. Marusic, and M. Palaniswami, ``Internet of Things
(IoT): A vision, architectural elements, and future directions,'' {\em \it Future
Generation Computer Systems,} vol. 29, nol. 7, pp. 1645-1660, 2013.





\bibitem{IGoiri01}
I. Goiri, K. Le, J. Guitart, J. Torres, and R. Bianchini, ``Intelligent Placement of Datacenters for Internet Services,'' in {\em \it Proc. IEEE ICDCS,} pp. 131-142, Minneapolis, MI, Jun. 2011.




\bibitem{HZhang03}
H. Zhang, Y. Xiao, S. Bu, D. Niyato, R. Yu, and Z. Han, ``Fog Computing in Multi-tier Data Center Networks: A Hierarchical Game Approach,'' in {\em \it Proc. IEEE ICC'16,} Kuala Lumpur, Malaysia, May 2016.

\bibitem{CGuo01}
C. Guo, G. Lu, and H. J. Wang, ``Secondnet: a Data Center Network Virtualization Architecture with Bandwidth Guarantees,'' in {\em \it Proc. of the 6th International Conference. ACM,} pp. 15, Graz, Austria, Sep. 2010.

\bibitem{AGan01}
A. Gandhi, M. Harchol-Balter, R. Das, and C. Lefurgy, ``Optimal Power Allocation in Server Farms,'' {\em \it ACM SIGMETRICS Perform. Eval. Rev.,} vol. 37, no. 1, pp. 157-168, Jun. 2009.

\bibitem{Roth01}
A. E. Roth, ``The Evolution of the Labor Market for Medical Interns and Residents: a Case Study in Game Theory,'' {\em \it Journal of Political Economy,} vol. 92, no. 6, pp. 991-1016, Dec. 1984.

\bibitem{ZLiu01}
Z. Liu, M. Lin, A. Wierman, S. H. Low, and L. L. Andrew, ``Geographical Load Balancing with Renewables,'' {\em \it ACM SIGMETRICS Perform. Eval. Rev.,} vol. 39, no. 3, pp. 62-66, Dec. 2011.

\bibitem{BAhlgren01}
B. Ahlgren, P. Aranda, P. Chemouil, S. Oueslati, L. Correia, H. Karl, M. Sollner, and A. Welin, ``Content, Connectivity, and Cloud: Ingredients for the Network of the Future,'' {\em \it IEEE Commun. Mag.,} vol. 49, no. 7, pp. 62-70, Jul. 2011

\bibitem{SIslam01}
S. Islam and J.-C. Gregoire, ``Network Edge Intelligence for the Emerging Next-Generation Internet,'' {\em \it Future Internet,} vol. 2, no. 4, pp. 603-623, Dec. 2010.

\bibitem{KChurch01}
K. Church, A. Greenberg, and J. Hamilton, ``On Delivering Embarrassingly Distributed Cloud Services,'' in {\em \it Proc. ACM HotNets,} pp. 55-60, Calgary, Canada, Oct. 2008.

\bibitem{VValancius01}
V. Valancius, N. Laoutaris, C. Diot, P. Rodriguez, and L. Massoulie,
``Greening the Internet with Nano Data Centers,'' in {\em \it Proc. ACM
CoNEXT,} pp. 37-48, Rome, Italy, Dec. 2009.


\bibitem{MBMobley01}
M. B. Mobley, T. Stuart, and Y. Andrew, ``Next-Generation Managed Services: A Window of Opportunity for Service Providers,'' {\em \it CISCO Technical Report,} 2009.

\bibitem{FBonomi01}
F. Bonomi, R. Milito, and J. Zhu, ``Fog Computing and its Role in the Internet of Things,'' in {\em \it Proc. of the first edition of the MCC workshop on Mobile cloud computing. ACM,} pp. 13-16, Helsinki, Finland, Aug. 2012.

\bibitem{Zjiang01}
Z. Jiang, D. S. Chan, M. S. Prabhu, P. Natarajan, H. Hao, and F. Bonomi, ``Improving Web Sites Performance Using Edge Servers in Fog Computing Architecture,'' {\em \it in Service Oriented System Engineering (SOSE), 2013 IEEE 7th International Symposium on,} pp. 320-323, Redwood City, CA, Mar. 2013.

\bibitem{Firdhous01}
M. Firdhous, O. Ghazali, and S. Hassan, ``Fog Computing: Will it be the Future of Cloud Computing?'' in {\em \it Proc. Third International Conference on Informatics and Applications (ICIA2014),} pp. 8-15, Kuala Terengganu, Malaysia, Jul. 2014.

\bibitem{Stojmenovic01}
I. Stojmenovic, and S. Wen, ``The Fog Computing Paradigm: Scenarios and Security Issues,'' in {\em \it Proc. Computer Science and Information Systems (FedCSIS), 2014 Federated Conference on,} Warsaw, Porland, 2014.


\bibitem{MChiang001}
M. Chiang, and T. Zhang, ``Fog and IoT: An Overview of Research Opportunities,'' in {\em \it in IEEE Internet of Things Journal,} to be published.

\bibitem{MAazam001}
M. Aazam, M. St-Hilaire, C. H. Lung, and I. Lambadaris, ``MeFoRE: QoE Based Resource Estimation at Fog to Enhance QoS in IoT,'' in {\em \it 2016 23rd International Conference on Telecommunications (ICT),} Thessaloniki, Greek, May 2016.

\bibitem{SFAbedin001}
S. F. Abedin, M. G. R. Alam, N. H. Tran, and C. S. Hong, ``A Fog Based System Model for Cooperative IoT Node Pairing Using Matching Theory,'' in {\em \it 2015 17th Asia-Pacific Network Operations and Management Symposium (APNOMS),} pp. 309-314, Busan, Korean, Aug. 2015.

\bibitem{NKGiang001}
N. K. Giang, M. Blackstock, R. Lea, and V. C. M. Leung, ``Developing IoT Applications in the Fog: A Distributed Dataflow Approach,'' in {\em \it 2015 5th International Conference on the Internet of Things (IOT),} pp. 155-162, Seoul, Korean, Oct. 2015.

\bibitem{MAazam002}
M. Aazam, and E. N. Huh, ``Fog Computing Micro Datacenter Based Dynamic Resource Estimation and Pricing Model for IoT,'' in {\em \it 2015 IEEE 29th International Conference on Advanced Information Networking and Applications,} pp. 687-694, Gwangiu, Korean, Mar. 2015.

\bibitem{MYannuzzi001}
M. Yannuzzi, R. Milito, R. Serral-Graci¨¤, D. Montero, and M. Nemirovsky, ``Key Ingredients in an IoT Recipe: Fog Computing, Cloud Computing, and more Fog Computing,'' in {\em \it 2014 IEEE 19th International Workshop on Computer Aided Modeling and Design of Communication Links and Networks (CAMAD),} pp. 325-329, Athens, Greek, Dec. 2014.


\bibitem{BMei001}
B. Mei, W. Cheng, and X. Cheng, ``Fog Computing Based Ultraviolet
Radiation Measurement via Smartphones,'' in {\em \it in 2015 Third IEEE Workshop
on Hot Topics in Web Systems and Technologies (HotWeb),} Washington
D.C., Nov. 2015.

\bibitem{ZHan001}
Z. Han, D. Niyato, W. Saad, T. Basar, and A. Hjorungnes, ``Game Theory in Wireless and Communication Networks: Theory, Models and Applications,'' in {\em \it Cambridge University Press,} UK, 2011.

\bibitem{CWang01}
C. Wang, Y. Yuan, C. Wang, X. Hu, and C. Zheng, ``Virtual Bandwidth Allocation Game in Data Centers,'' in {\em \it Proc. 2012 IEEE International Conference on Information Science and Technology,} pp. 682-685, Hubei, China, 2012.

\bibitem{MWu001}
M. Wu, D. Ye, S. Tang, and R. Yu, ``Collaborative Vehicle Sensing in Bus Networks: A Stackelberg Game Approach,'' {\em \it 2016 IEEE/CIC International Conference on Communications in China (ICCC),} Chengdu, China, Aug. 2016.

\bibitem{YWang001}
Y. Wang, X. Lin and, M. Pedram, ``A Stackelberg Game-Based Optimization Framework of the Smart Grid With Distributed PV Power Generations and Data Centers,'' {\em \it in IEEE Transactions on Energy Conversion,} vol. 29, no. 4, pp. 978-987, Dec. 2014.

\bibitem{HWang001}
H. Wang, Y. Zhao, and H. Guan, ``On Pricing Schemes in Data Center Network with Game Theoretic Approach,'' {\em \it 2014 23rd International Conference on Computer Communication and Networks (ICCCN),} Shanghai, China, Aug. 2014.

\bibitem{BYang01}
B. Yang, Z. Li, S. Chen, T. Wang, and K. Li, ``Stackelberg Game Approach for Energy-aware Resource Allocation in Data Centers,'' {\em \it in IEEE Transactions on Parallel and Distributed Systems ,} to appear.

\bibitem{HZhang01}
H. Zhang, M. Bennis, L. A. DaSilva, and Z. Han, ``Multi-leader Multifollower Stackelberg Game among Wi-Fi, Small Cell and Macrocell Networks,'' in {\em \it Proc. IEEE Globecom'14,} Austin, TX, Dec. 2014.

\bibitem{HZhang02}
H. Zhang, Y. Xiao, L. X. Cai, L. Song, N. Dusit, and Z. Han, ``Hieratical Competition for LTE Unlicensed using Stackelberg Game and Bargaining,'' in {\em \it Proc. IEEE Globecom'15,} San Diego, CA, Dec. 2015.




\bibitem{LMAusubel02}
L. M. Ausubel, ``An Efficient Dynamic Auction for Heterogeneous Commodities,'' {\em \it
American Economic Review,} vol. 96, no. 3, pp. 602-629, Jun. 2006.

\bibitem{GDemange01}
G. Demange, D. Gale, and M. Sotomayor, ``Multi-item Auctions,'' {\em \it
Journal of Political Economy,} vol. 94, no. 4, pp. 863-872, Aug. 1986.

\bibitem{DMishra01}
D. Mishra and R. Garg, ``Descending Price Multi-item Auctions,'' {\em \it Journal
of Mathematical Economics,} vol. 42, no. 2, pp. 161-179, Apr. 2006.



\bibitem{RNarayanam01}
R. Narayanam, ``A Game Theoretic Approach to Identify Critical Components in Networked Systems,'' in {\em \it Proc. 2012 Annual SRII Global Conference,} pp. 515-524, San Jose, CA, 2012.

\bibitem{MKearns01}
M. Kearns, ``Graphical Games,'' {\em \it  Algorithmic Game Theory, book chapter, Cambridge University Press,} pp. 159-180, 2007.



\end{thebibliography}
\end{document}